\documentclass[10pt, conference, letterpaper]{IEEEtran}
\IEEEoverridecommandlockouts
\usepackage{cite}
\usepackage{amsmath}
\usepackage{amssymb,amsfonts}
\usepackage{mathrsfs}
\usepackage{graphicx}
\usepackage{textcomp}
\usepackage{xcolor}
\def\BibTeX{{\rm B\kern-.05em{\sc i\kern-.025em b}\kern-.08em
    T\kern-.1667em\lower.7ex\hbox{E}\kern-.125emX}}
    
\usepackage[numbers,sort&compress]{natbib}
 \usepackage{enumitem} 
 
\usepackage{cleveref} 
\usepackage{amsthm}

\newtheorem{cor}{Corollary}
\crefname{cor}{corollary}{corollaries}
\newtheorem{pro}{Proposition}
\crefname{pro}{proposition}{propositions}

\usepackage{algorithm}               
\usepackage{algorithmic}             
\usepackage{multirow}                
\usepackage{graphics}
\usepackage{epsfig}

\usepackage{url} 
\makeatletter
\def\url@leostyle{%
  \@ifundefined{selectfont}{\def\UrlFont{\sf}}{\def\UrlFont{\small\ttfamily}}}
\makeatother
\usepackage{subfigure} 
\usepackage{multirow}                
\usepackage{booktabs}

\usepackage{threeparttable} 

\begin{document}

\title{Consistent User-Traffic Allocation and Load Balancing in Mobile Edge Caching}

\author{\IEEEauthorblockN{Lemei Huang\IEEEauthorrefmark{1},
Sheng Cheng\IEEEauthorrefmark{1},
Yu Guan\IEEEauthorrefmark{1}\IEEEauthorrefmark{2},
Xinggong Zhang\IEEEauthorrefmark{1}\IEEEauthorrefmark{2} and
Zongming Guo\IEEEauthorrefmark{1}\IEEEauthorrefmark{2}}
\IEEEauthorblockA{\IEEEauthorrefmark{1}Institute of Computer Science Technology, Peking University, China, 100871}
\IEEEauthorblockA{\IEEEauthorrefmark{2}PKU-UCLA Joint Research Institute in Science and Engineering}
\IEEEauthorblockA{\{milliele, chaser\_wind, shanxigy, zhangxg, guozongming\}@pku.edu.cn}}

\maketitle

\begin{abstract}
Cache-equipped Base-Stations (CBSs) is an attractive alternative to offload the rapidly growing backhaul traffic in a mobile network. New 5G technology and dense femtocell enable one user to connect to multiple base-stations simultaneously. Practical implementation requires the caches in BSs to be regarded as a cache server, but few of the existing works considered how to offload traffic, or how to schedule HTTP requests to CBSs. In this work, we propose a DNS-based HTTP traffic allocation framework. It schedules user traffic among multiple CBSs by DNS resolution, with the consideration of load-balancing, traffic allocation consistency and scheduling granularity of DNS. To address these issues, we formulate the user-traffic allocation problem in DNS-based mobile edge caching, aiming at maximizing QoS gain and allocation consistency while maintaining load balance. Then we present a simple greedy algorithm which gives a more consistent solution when user-traffic changes dynamically. Theoretical analysis proves that it is within 3/4 of the optimal solution. Extensive evaluations in numerical and trace-driven situations show that the greedy algorithm can avoid about 50\% unnecessary shift in user-traffic allocation, yield more stable cache hit ratio and balance the load between CBSs without losing much of the QoS gain.
\end{abstract}


\vspace{-0.2cm}

\section{Introduction}

Mobile traffic grows explosively in recent years and is likely to increase seven-fold between 2017 and 2022\cite{vni}. Cache-equipped Base-Station (CBS) is an attractive alternative to offload the backhaul traffic, especially for large scale multimedia services. Besides, it can improve user Quality of Service (QoS) significantly by reducing content-fetching latency \cite{cbs2}. 

With the development of mobile technologies (e.g.\ 5G), small BSs (e.g.\ Femto BSs) are densely distributed in an network area and a user can be associated to multiple BSs in the neighborhood (rather than just the closest one), as depicted in FemtoCacing \cite{femto2012}. Therefore, it is necessary to decide how user requests should be routed, a.k.a.\ user-traffic allocation, and CBSs can serve user traffic collaboratively.
Plenty of works have studied user-traffic allocation problems in collaborative mobile edge caching \cite{femto2012,delay1,infocom18}, while few of them consider the practical implementations of a CBS. To incrementally deploy cache in a BS without modifying the rest of the Internet infrastructure, caches need to support application-level protocols, which means equipped cache in a BS should still be treated as a \emph{cache server} \cite{practical}
\footnote{The majority of mobile traffic, especially video traffic, are transmitted by application-level protocol, e.g.\ HTTP, RTP, etc.}.
DNS resolution has been widely used in user-traffic allocation to cache servers \cite{akamai2012,dnsoverview,dns2}. DNS servers respond with the location of caches or source server according to its traffic allocation policy when users issue content requests. However, DNS-based user-traffic allocation faces two main challenges:

\begin{figure}[tb]
\centering
\includegraphics[width=0.4\textwidth]{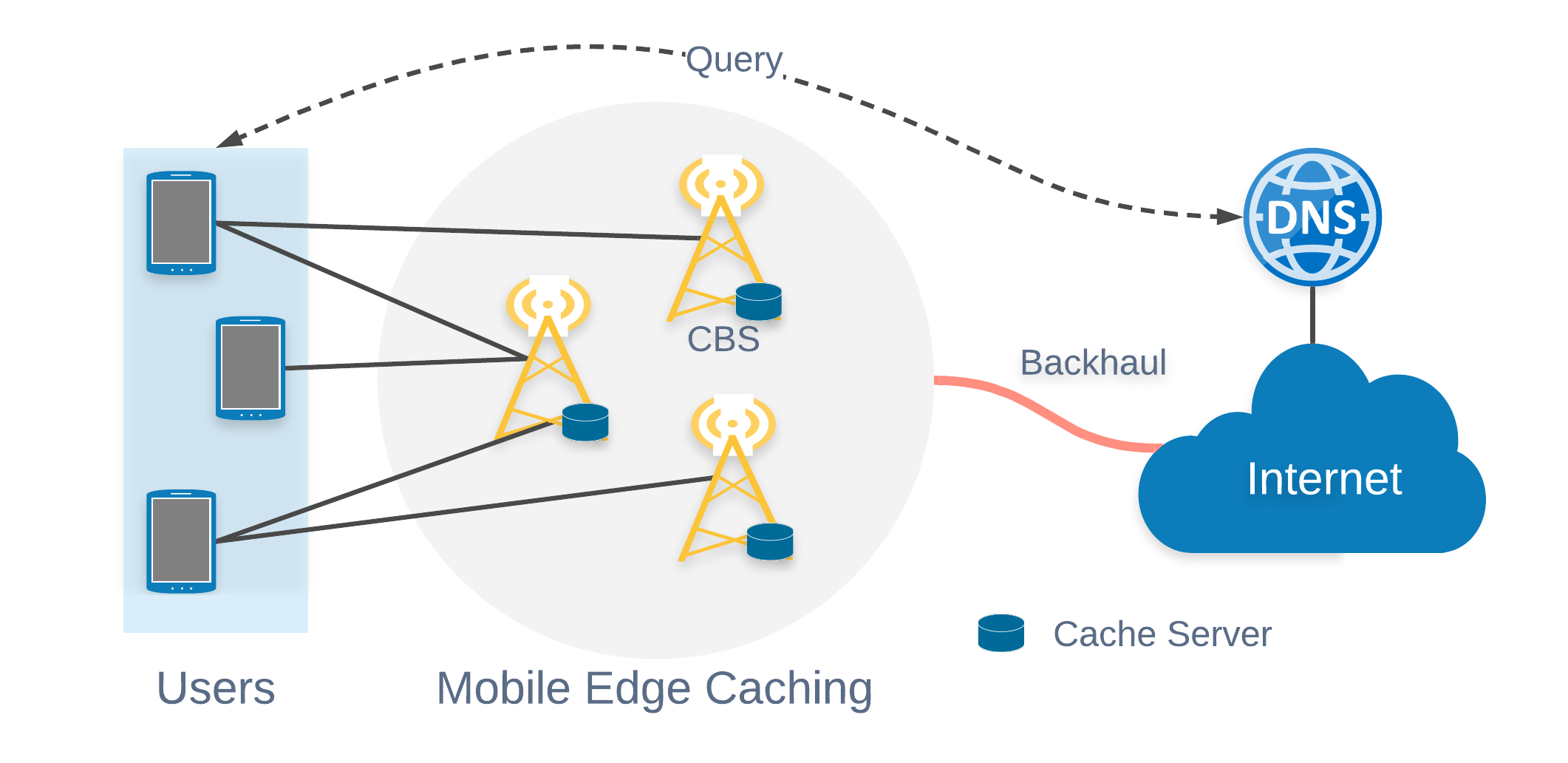}
\caption{\textbf{Mobile edge caching: }Cache-equipped base stations (CBSs) are densely distributed in a network area. A user can be associated with multiple CBSs close enough. Caches equipped on CBSs are served as cache servers\cite{practical}, user traffic is allocated to them by DNS resolution. 
\label{framework}}
\end{figure}

\textbf{Allocation granularity. }DNS-based traffic allocation is realized by carefully determining an IP list of available caches. For example, if the traffic allocated to cache $A$ and $B$ is $3:1$, the IP list may contain 3 identical IP addresses of $A$ and 1 IP address of $B$ \cite{rr}. The size of this IP list, however, is often limited, as DNS responses are normally propagated in only one UDP packet \cite{dnsoverview} with a max-length of 512 Bytes \cite{512}
\footnote{Although TCP protocol or EDNS option can be used to expand the packet length, it requires extra support of the application.}.
Therefore, one of the practical concerns omitted in many theoretical models is that user traffic cannot be distributed in an arbitrary proportion \cite{delay1,infocom18,discrete}. Although the solution to these models can be adjusted to applied in DNS resolution, it may lead to a suboptimal policy. We believe it is better to take into account the allocation granularity when formulating the user-traffic allocation problem.

\textbf{Allocation Consistency. }Although most of the existing works \cite{femto2012,delay1,infocom18,discrete} aim to jointly solve the content placement and user-traffic allocation problem, we suppose that reactive caching (e.g.\ LRU cache) is applied in CBSs to cope with dynamic user traffic and content popularity in DNS-based scenarios\footnote{We make this assumption for two reasons. First, in most the cases, DNS server knows only the domain name of the requested content rather than the detailed URL. Second, if content placement remains unchanged, cache utility will degrade when content popularity changes, and frequently recalculating content placement policy is computationally unacceptable.}.
The cache of each CBS is large enough to cache most of the popular contents in the user traffic \cite{femto2012}, but lacking the capacity to cache all the contents in the category. In a mobile network, both user traffic and user location change dynamically. Therefore, user-traffic allocation policy needs to be adaptive and ever-changing (see \Cref{eva}).
"Inconsistent" traffic allocation may result in severely high cache miss ratio, which reduces the utility of a CBS. For example, suppose traffic belonging to \verb|Y.com| is reallocated from cache $A$ to $B$. $B$ is very unlikely to cache contents under \verb|Y.com|, since it does not serve \verb|Y.com| before. Therefore, $B$ must fetch these contents from the Internet, which may lead to cache misses, and cache them to serve subsequent user requests for \verb|Y.com|.
It becomes harmful when the strategy is so "sensitive" that it comes up with a very "inconsistent" (i.e.\ totally different) traffic allocation policy just to attain a trivial estimated gain in QoS after recalculation.

In this paper, we focused on the scenarios where users can be associated with multiple close enough CBSs, each CBS equipped with a cache server, and user traffic is allocated by DNS resolution (\Cref{framework}). 
In summary, the main contributions of this paper are as follows:
\begin{itemize}[leftmargin=*]
\item We present and formulate user-traffic allocation (UTA) problem (\Cref{dlsbformulation}), the problem of how can user traffic be allocated to CBSs considering the allocation granularity of DNS. In UTA problem, we want to maximize QoS gain, balance the load, and keep the allocation policy consistent.
\item We prove that the formulation can be transformed equivalently into a typical combinatorics optimization and then presented a simple greedy algorithm with 3/4 approximation ratio (\Cref{theoretical}).
\item We provide extensive evaluations in both numerical and trace-driven situations (\Cref{eva}). The results show that our algorithm yields more consistent traffic allocation while maintaining the QoS gain and balancing the load.
\end{itemize}

\section{Related Work}
\label{relatedwork}

\textbf{User-traffic allocation in mobile edge caching. }
In \cite{femto2012}, traffic allocation depends on the placement of the content; user requests can be allocated to any neighboring CBS caching the content. \cite{delay1} further examines the delay between users and CBSs to optimize user QoS. However, as CBSs are expected an unlimited traffic capacity, user traffic is usually allocated to the nearest CBS with the requested content, which is likely to overburden a CBS. \cite{band1,band2,band3,discrete} allocate user traffic considering the upper bound of traffic that a CBS can serve and avoid QoS degradation when a CBS caching popular contents is submerged by user requests and thus overburdened. In their formulation, however, no matter whether the variables of traffic allocation are continuous \cite{band1} or not \cite{discrete}, user traffic could be distributed in arbitrary portion, which can not be implemented due to the allocation granularity of DNS resolution.

\textbf{Consistent user-traffic allocation. }
Consistent hashing \cite{consistenthashing} has long been used by cache networks to balance server load with minimum allocation changes. Jiang et.al. formulated a traffic allocation problem in CDNs comprising distributed caches (e.g. set-top boxes) in \cite{discdn}. Their algorithm updates traffic allocation policy smoothly. While these algorithms suit clusters consisting of caches with identical characteristics, our work targets a different problem where CBSs are often heterogeneous. Even in the context of wired content delivery, few of the works considered allocation consistency in such a problem. To the best of our knowledge, in 2015, Akamai transformed the load balancing problem into a variant of the stable marriage one and presented their generalized Gale-Shapley (GGS) algorithm for traffic allocation \cite{akamai2012}. They pointed out the importance of allocation consistency but did not present a detailed solution. GSS algorithm can allocate user traffic when caches are heterogeneous. However, optimality of this solution lacks theoretical analysis as well.


\textbf{Practical implementations of CBS. }
An example of practical implementations of CBS is presented in \cite{practical}. The equipped cache of a Base Station (BS) node (i.e. eNodeB) is installed as a cache server in its kernel, and modification is made on the protocol stack of BS. In this way, the BS node can act as not only a gateway but also a router that can forward IP packets. IP packets whose destination is the equipped cache server are captured and served by cache, others forwarded to the Serving Gateway (S-GW) in Evolved Packet Core (EPC). Thus, cache on BS can support application-level data flow. Content requests are allocated to each cache by DNS resolution. However, user-traffic allocation or other optimization problems in mobile edge caching remain understudied in \cite{practical}.



\section{Network Model and Problem Statement}
\label{dlsbformulation}

\subsection{Wireless Environment and Cache Network}

\subsubsection*{CBSs and User Association}
Each CBS covers a communication range with a specific radius in a 2D plane. Users within such range can be associated with it. Thanks to the dense distribution of CBSs, users can have multiple association alternatives when CBSs' communication ranges overlap. We do not consider multi-source downloading, and therefore each content request can only be assigned to at most one CBS.

Let there be a group of CBSs $\mathbf{M}$ in the network area. The amount of user traffic that a CBS can handle has an upper bound, a.k.a.\ \emph{capacity}. A CBS are probably overburdened when it receives more traffic than its capacity. We use $c_j$ to denote the capacity of each CBS $j\in\mathbf{M}$.

\subsubsection*{User Traffic}
To cope with the high computational complexity caused by user quantity in the network area, users is grouped into User Groups (UGs) $\mathbf{U}$. In each UG, it is assumed that all users locate in neighboring locations and experience the same radio conditions regarding fading and interference. Therefore, they share the same group of associable CBSs. For example, users in the same building could be grouped together.
User traffic from each UG $u\in\mathbf{U}$ are grouped before allocation to CBSs. Let $\mathbf{S}$ be the category of domain names. Thus, user traffic can be divided into different flows. Each flow $i=\langle u,s \rangle \in \mathbf{F}=\mathbf{U}\times\mathbf{S}$ comprising content requests from UG $u$ under domain name $s$. The amount of user traffic in flow $i$ is denoted by $\lambda_i$ (see \verb|Flow| layer in \Cref{overview}).

We denote the connectivity between users and CBSs by $t_{ij}\in\{0,1\}$. For $i=\langle u,s\rangle$, $t_{ij}=1$ when UG $u$ can be associated with CBS $j$ and otherwise $t_{ij}=0$. $\{t_{ij}\}$ are also constant variables identical for all the sub-flows of flow $i$.

\subsubsection*{Mapping Units}
Considering the granularity in DNS-based traffic allocation, we assume that each flow $i$ is divided evenly into $|\mathbf{K_i}|$ sub-flows a prior, where each sub-flow carries an equal amount of traffic $\lambda_0$, applying certain rounding methods. Since the division is uniform, we can assume that in the same flow, all the sub-flows are identical regarding the category of requested contents, user group, and domain name. The long tail effect of content popularity distribution may generate a lot of small flows whose amount of traffic is less than $\lambda_0$. In practice, those flows in the same or nearby regions can be merged into one flow with the amount approximate to $\lambda_0$ to take full advantage of the CBS capacity. Suppose the maximum length of IP list is $\iota $, $\lambda_0$ can be set up as $1/\iota $ of the maximal traffic amount of all the flows, which ensures that $|\mathbf{K_i}|\le \iota$, thus enabling a realistic allocation policy in DNS-based load balancing. 
We refer to a sub-flow $\langle i, k\rangle$ as a mapping unit, for $i\in\mathbf{F}$ and $k\in\mathbf{K_i}$. $k$ can be regarded as the identifier of each sub-flow, as depicted in \Cref{overview}. There are other methods to divide flows into the granularity suitable for DNS server allocation, which is beyond our discussion as our main purpose in this paper is to present an optimization model with considerations of practical limitations.

Let $x_{ij}^k (i\in\mathbf{F},j\in\mathbf{M},k\in\mathbf{K_i})$ be the 0-1 variable that indicates if mapping unit $\langle i, k\rangle$ should be allocated to CBS $j$ when $x_{ij}^k=1$, or not otherwise. Then $\mathbf{X}=[x_{ij}^k]$ is our traffic allocation policy.

Notice that user requests that are not assigned a CBS can connect to any neighboring BS --not necessarily a CBS-- to fetch the content from the Internet.

\begin{figure}[tb]
\centering
\includegraphics[width=0.45\textwidth]{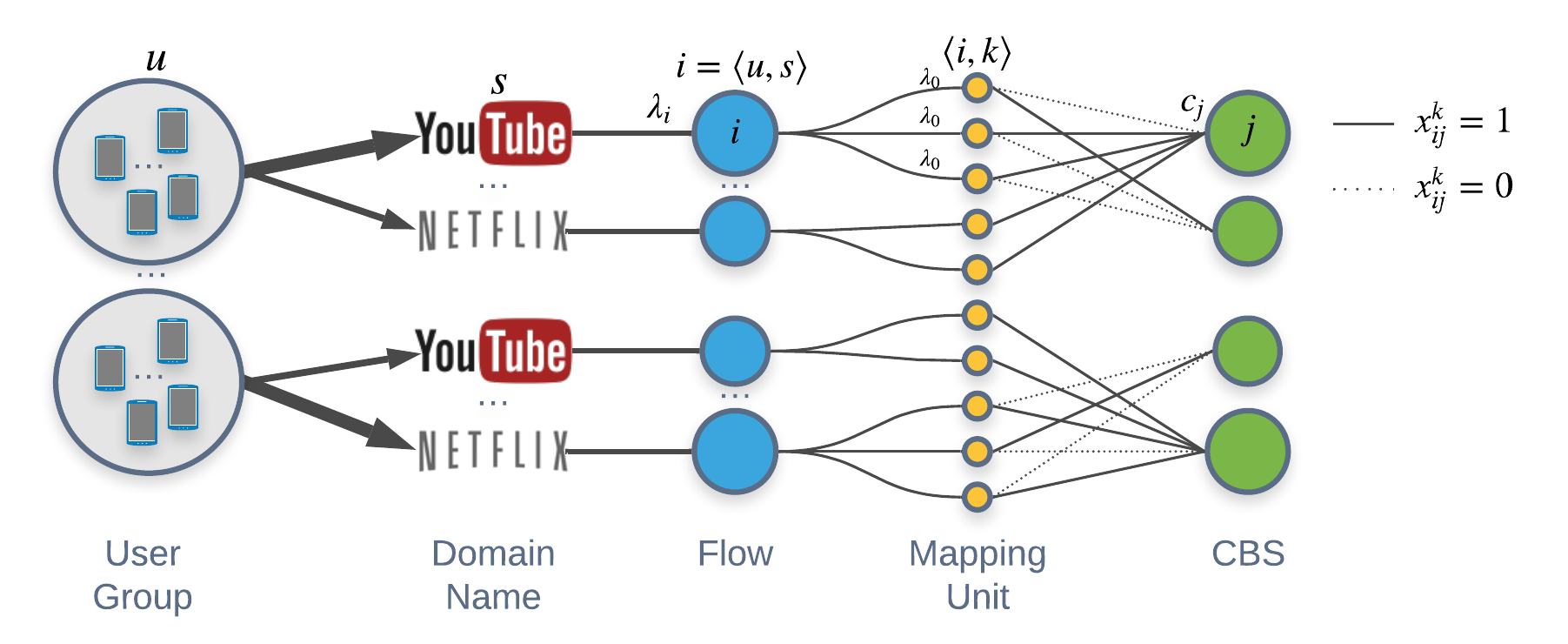}
\caption{\textbf{UTA problem:} user traffic are partitioned into flows based on \emph{(user group, domain name)}; flows are divided into mapping units to be allocated to CBSs. Since a user can be associated with multiple CBSs, each mapping unit can have multiple alternative CBSs to be matched (solid and dotted lines), and can be allocated to some of them (solid lines).
\label{overview}}
\end{figure}

\subsection{Optimization Goal}
\label{optgoal}

\subsubsection*{QoS Gain}
In our problem, we suppose that QoS depends mainly on the content fetching delay. When a user receives the content from caches in CBSs rather than the Internet, there is QoS gain because content fetching delay is saved. We denote the QoS gain of allocating a user request of mapping unit $\langle i,k\rangle$ to CBS $j$ by $g_{ij}$, and we assume $g_{ij}\ge 0$. Because of the homogeneity of sub-flows, all the sub-flows within the same flow have the same estimated QoS gain.

The estimation of $g_{ij}$ involves multiple concerns, including distance between a user and CBS, backhaul delay, etc.. When the cache in a CBS has a high capacity but unfortunately small storage, user requests routed to it may suffer frequent cache miss, which is likely to damage the QoS gain as well. In our problem, we regard $g_{ij}$ as a constant factor and do not focus on the estimation of it.

Now the total QoS gain can be quantified as:
\begin{equation}
\label{qx}
G(\mathbf{X})=\sum_{j\in\mathbf{M}}\sum_{i\in\mathbf{F}}g_{ij}\sum_{k\in\mathbf{K_i}}\lambda_0x^k_{ij}
\end{equation}

\subsubsection*{Load Balancing}
To balance the load, we first ensure each CBS receives user traffic no more than its capacity:
\begin{equation}
L_j(\mathbf{X})=\sum_{i\in\mathbf{F}}\sum_{k\in\mathbf{K_i}}\lambda_0x^k_{ij}\leq c_j
\end{equation}
, where $L_j(\mathbf{X})$ quantifies user traffic that CBS $j$ receives, which also represents the load of CBS $j$.

Second, considering that balancing spare capacity between CBSs can not only help to offload popular CBSs and increase the utilization of light-loaded CBSs, but also reduce the risk of overload when "flash crowd" occurs, we maximize the spare capacity fairness between CBSs by maximizing $$
B(\mathbf{X})=\sum_{j\in\mathbf{M}}H(c_j-L_j(\mathbf{X}))
$$. $H(\cdot)$ can be any decreasing concave function on $[0,\max(c_j)]$, which encourages a mapping unit to be allocated to a light-loaded CBS for higher gain. For instance, we use $H(v) = -v^2$, and thus we have:
\begin{equation}
B(\mathbf{X})=-\sum_{j\in\mathbf{M}}(c_j-L_j(\mathbf{X}))^2
\end{equation}

\subsubsection*{Allocation Consistency}
A CBS can suffer severe cache miss and consume more backhaul bandwidth if the domain names it serves change a lot after the recalculation of allocation policy. Therefore, our traffic allocation policy needs to be less "sensitive" and more consistent.

We denote whether CBS $j$ serves mapping unit(s) under the domain name of flow $i$ by $w_{ij}$, which indicates the favorability of mapping $\langle i,k\rangle$ to $j$. We assume $w_{ij}\ge 0$ as well. Note that the identifier $k$ of sub-flow makes no difference to $w_{ij}$. Specifically, the estimation of $w_{ij}$ is based on the previous traffic allocation policy $\mathbf{\hat{X}} = [\hat{x}^k_{ij}]$. Suppose $\mathbf{F_s}$ consists of mapping units under the same domain name $s$, and then $w_{ij}$ can be defined as:
$$
w_{ij} = \left\{\begin{matrix}
1 &,\text{ if }i=\langle *,s\rangle\text{ and }\exists i'\in \mathbf{F_s}, \hat{x}^k_{i'j}=1 \\ 
0 &,\text{ otherwise}
\end{matrix}\right.
$$
We can then maximize $W(\mathbf{X})$ to maximize the consistency of traffic allocation policy $\mathbf{X}$:
\begin{equation}
W(\mathbf{X})=\sum_{j\in\mathbf{M}}\sum_{i\in\mathbf{F}}\sum_{k\in\mathbf{K_i}}w_{ij}x^k_{ij}
\end{equation}
Similar to mapping units, there may also be other methods to quantify $w_{ij}$, which is beyond the scope of this paper.

\subsection{User-traffic Allocation Problem}
All in all, we present user-traffic allocation (UTA) problem in a DNS-based mobile edge caching that aims at maximizing QoS gain, allocation consistency as well as maintaining load balancing. We can formulate UTA problem as: 

\begin{subequations}
\label{p1}
\begin{alignat}{2}
\max\quad & F(\mathbf{X})=\mu_1 G(\mathbf{X})+\mu_2 B(\mathbf{X})+\mu_3 W(\mathbf{X}) \label{LM:1}\\
\mbox{s.t.}\quad &L_j(\mathbf{X})\leq c_j, \forall j\in\mathbf{M}\label{LM:2}\\
&\sum_{j\in\mathbf{M}}x^k_{ij}\le 1, \forall i\in\mathbf{F},k\in\mathbf{K_i}\label{LM:3}\\
&x^k_{ij}\le t_{ij}, \forall i\in\mathbf{F},k\in\mathbf{K_i}\label{LM:5}\\
&x_{ij}^k\in\{0,1\}, \forall i\in\mathbf{F},j\in\mathbf{M},k\in\mathbf{K_i}\label{LM:4}
\end{alignat}
\end{subequations}

\eqref{LM:1} is the objective function, where $\mu_1$, $\mu _2$ and $\mu_3$ are constant variables to balance the trade-off among the three goals in this optimization model, \emph{QoS gain}, \emph{load balancing} and \emph{allocation consistency}.
Since we do not assume multi-source downloading, \eqref{LM:3} indicates that each mapping unit should be allocated to at most one CBS.
Finally, \eqref{LM:2}, \eqref{LM:5}, \eqref{LM:4} are capacity, connectivity and integrality constraints, respectively.

\section{Algorithm with Optimality Guarantee}
\label{theoretical}

The formulation of UTA problem \eqref{p1} is a problem of NP-hard quadratic integer programming, which calls for a computationally efficient approximate algorithm.
In this section, we prove that our model \eqref{p1} is equivalent to maximizing a monotone submodular function which is subject to matroid constraints. We then propose a simple and elegant greedy algorithm with the considerable approximation ratio of 3/4 to solve this problem.

\textbf{Properties of~\eqref{p1} (abstract). }The integrality constraint~\eqref{LM:4} enables that every cache decision $\mathbf{X} = [x^k_{ij}]$ can be written as a set $A\subset\{f^k_{ij}|i\in\mathbf{N},j\in\mathbf{M},k\in\mathbf{K_i}\}$,where $x^k_{ij}=1\Leftrightarrow f^k_{ij}\in A$. Thus, the constraints of~\eqref{p1} can be written as matroid constraints, according to the definition of partition matroids \cite{partionmatroid}. Moreover, the objective function~\eqref{LM:1} can be written as a set function  \cite{partionmatroid} which is a monotone submodular function.  Thus, \eqref{p1} is equivalent to maximizing a monotone submodular function which is subject to matroid constraints. Due to space constraints, detailed proof is in \Cref{detailed}.

Fisher et al.\cite{greedy} presents a simple and common \emph{greedy algorithm} to approximately solve the optimization problem that maximizing monotone submodular function subject to matroid constraints, with specific optimality guarantees. Before introducing the algorithm, we define the marginal value of allocating mapping unit $\langle i, k \rangle$ to CBS $j$ as
$$m_{ij}^k(X)=F(\mathbf{X}|x_{ij}^k=1)-F(\mathbf{X}|x_{ij}^k=0)$$
, where $x_{ij}^k=1$ refers to the new matrix generated by
changing $x_{ij}^k$ of $\mathbf{X}$ from 0 to 1. Then the offline algorithm is described in \Cref{greedy}, which keeps on greedily choosing a tuple $(i,j,k)$ with highest marginal value under constraints \eqref{LM:2}-\eqref{LM:5}, and then allocating $\langle i, k \rangle$ to $j$, i.e.\ let $x_{ij}^k=1$. \Cref{optimality} has proven the solution obtained by  \Cref{greedy} yields a 3/4 approximation.

\begin{algorithm}[htb]         
    \caption{The greedy algorithm} 
    \label{greedy} 
    \begin{algorithmic}[1]                
    \STATE Initializing: $\mathbf{X}\leftarrow \{\mathbf{0}\}$;
    \STATE $\mathbf{C}\leftarrow \{(i,j,k)\ |\ x_{ij}^{k}=0$\\
    \ \ \ \ \ \ \ \ \ \ and $\sum_{i'\in\mathbf{F}}\sum_{k'\in\mathbf{K_{i'}}}\lambda_0x^{k'}_{i'j}\leq c_j-\lambda_0$\\
    \ \ \ \ \ \ \ \ \ \ and $\sum_{j'\in\mathbf{M}}x^k_{ij'}\le 0$\\
    \ \ \ \ \ \ \ \ \ \ and $x^{k}_{ij}\le t_{ij}\}$
    \WHILE{$\mathbf{C}\neq\emptyset$}\label{check}
    \STATE $(i_0, j_0, k_0)\leftarrow argmax_{(i, j, k)\in \mathbf{C}}m_{ij}^{k}(\mathbf{X})$;
    \STATE $\mathbf{X}\leftarrow\mathbf{X}|x_{i_0j_0}^{k_0}=1$;     
    \ENDWHILE
    \RETURN $\mathbf{X}$;                
\end{algorithmic}
\end{algorithm}




\section{Evaluation}
\label{eva}
In this section, we carry out both numerical and trace-driven evaluations to explore the influence of different factors. We compare our approximation algorithm with other baselines in terms of QoS gain, load balancing and cache miss ratio.

\subsection{Experimental Setup}
\label{setup}
\subsubsection*{Wireless Environment}

We consider a $500\times500 m^2$ wireless network area. The positions of CBSs follow the Poisson Point Process (PPP) with a density of $80\frac{CBS}{{km}^2}$.
To simulate the situation of heterogeneous CBSs, we assume that the radius of CBS's communication range is randomly chosen from $150m$ to $300m$.
We fix the total capacity \verb|CAPACITY|. As CBS that covers larger communication range is likely to associate with more users and handle more user traffic, we let capacity of each CBS be proportional to its coverage radius. When CBS receives more user traffic than its capacity, we simply assume that it just rejects to serve the exceeded part of requests. However, in practice, overloading may lead to more severe consequences such as damage to the function of the BSs. User groups are distributed with a density of 40 UGs per ${km}^2$.

All the content files are equally sized as 30MB. We set the cache size of each CBS to be (i) 60GB for numerical evaluations (ii) 3GB for trace-driven evaluations. Caches all use Least Recently Used (LRU) policy to do content replacement.
We use $g_{ij}=100e^{-d_{ij}/500}$ to estimate the QoS gain, where $d_{ij}$ refers to the distance between the location of flow $i$'s user group $u$ and CBS $j$.

\begin{figure}[tbp]
\centering
\subfigure[Average CHRD of algorithms (\%)
\label{chrd}]{
\centering
\includegraphics[width=0.3\textwidth]{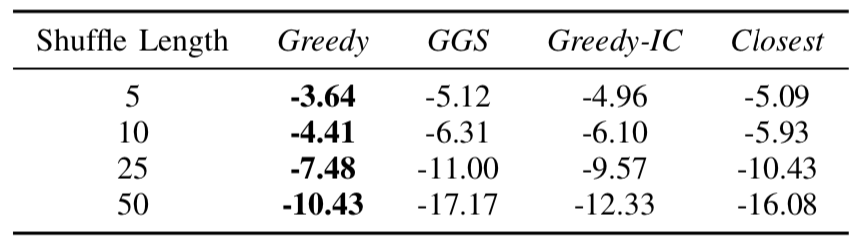}
}
\subfigure[Cache Hit Ratio
\label{rchr}]{
\centering
\includegraphics[width=0.15\textwidth]{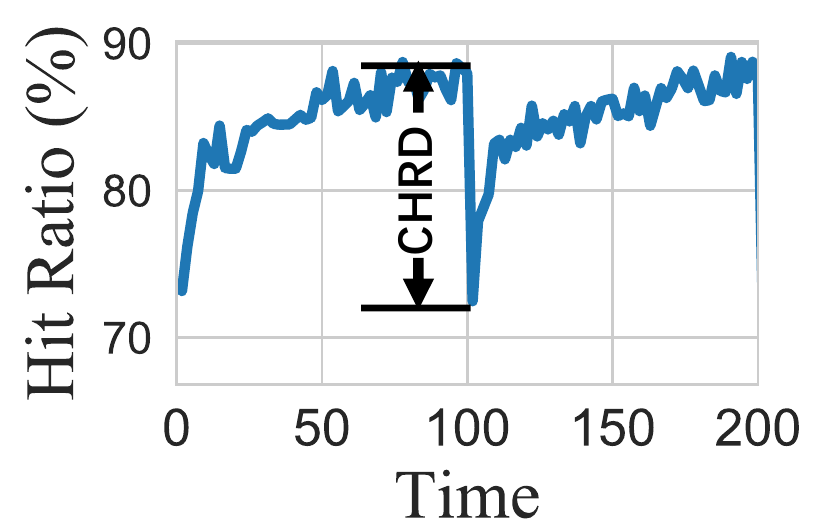}
}
\caption{\emph{Scenario 2: }the real-time average cache hit ratio of all the caches suffers an immediate, drastic decrease after the regeneration of traffic allocation policy due to "inconsistent" allocation, as many contents need to be replaced. \emph{Greedy} have lower CHRD because of its consistency.
\label{dyn2chrd}}
\end{figure}

\begin{figure*}[tb]
\centering
\subfigure[QoS gain
\label{dyn1q}]{
\centering
\includegraphics[width=0.3\textwidth]{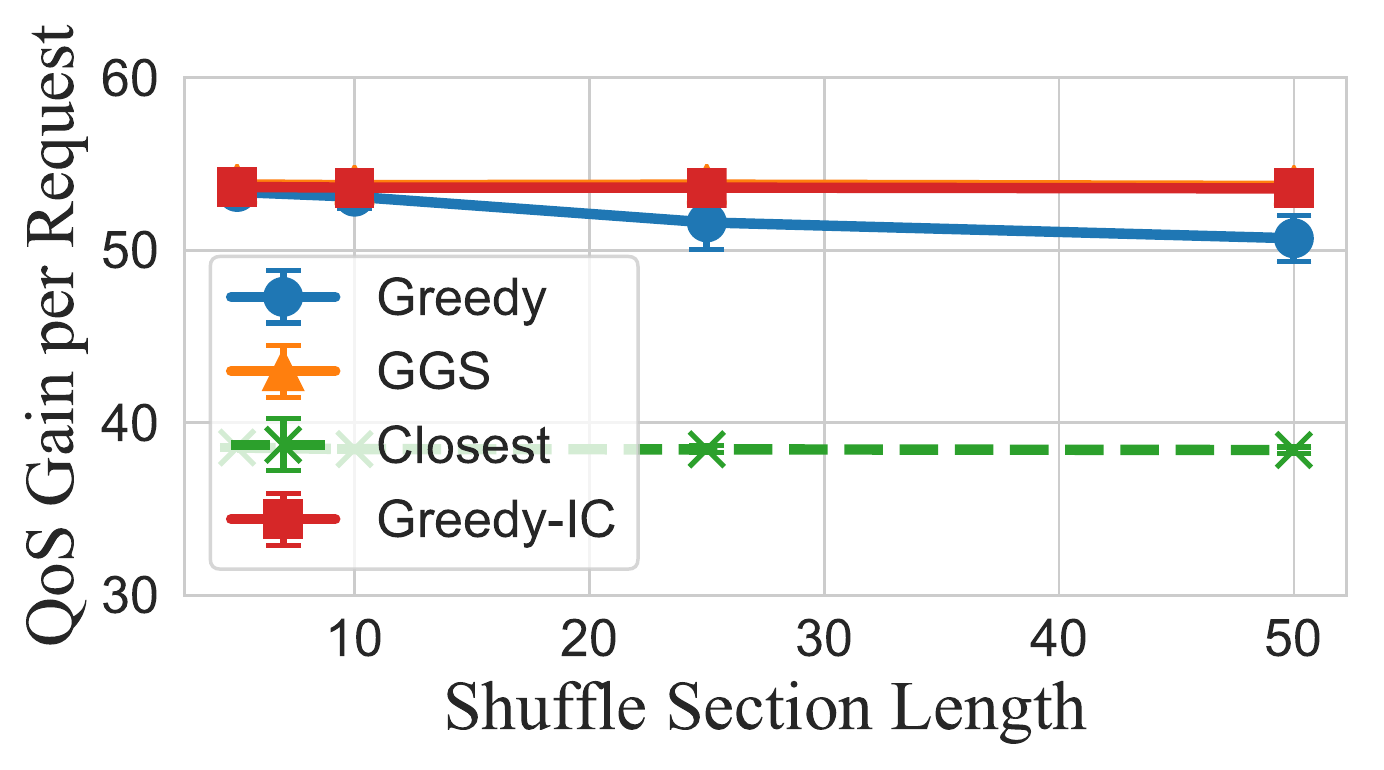}
}
\subfigure[Consistency
\label{dyn1diff}]{
\centering
\includegraphics[width=0.3\textwidth]{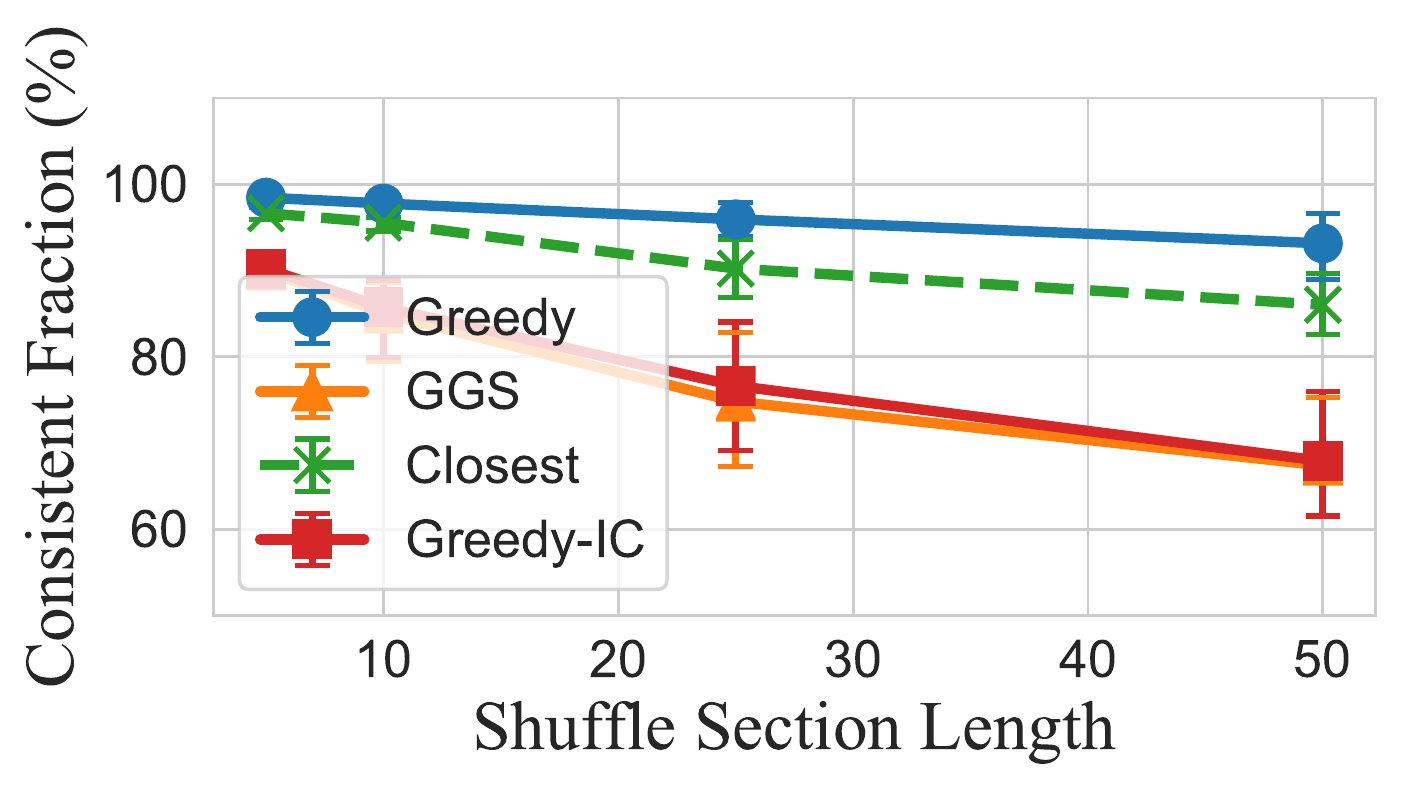}
}
\subfigure[Cache Miss Ratio
\label{dyn1miss}]{
\centering
\includegraphics[width=0.3\textwidth]{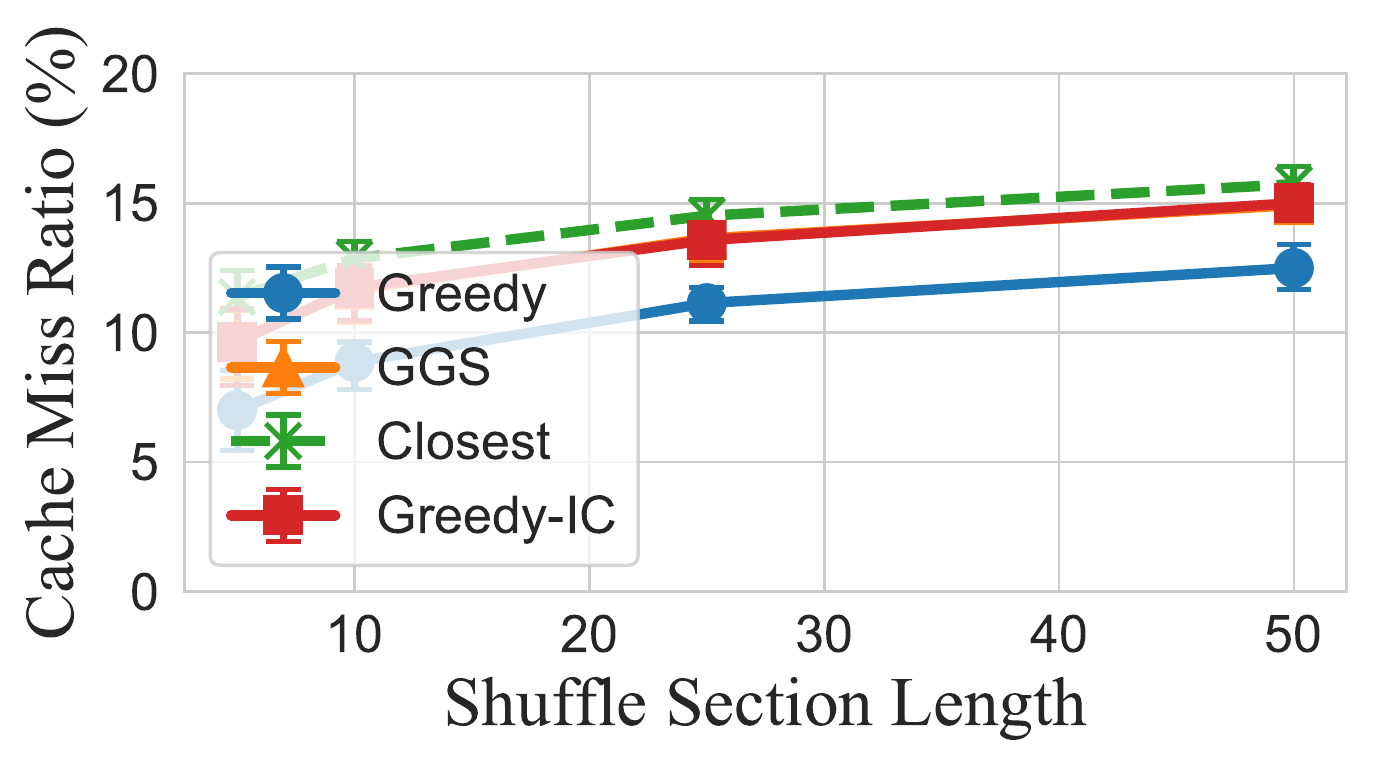}
}
\caption{\emph{Scenario 1: }we divide domain names into sections and randomly shuffle the popularity of domain names within each section. Content Popularity changes more when section length grows. \emph{Greedy} performs more consistent and yields less cache miss ratio without losing much QoS gain.
\label{dynamic1}}
\end{figure*}
\begin{figure*}[tb]
\centering
\subfigure[QoS gain
\label{traceq}]{
\centering
\includegraphics[width=0.3\textwidth]{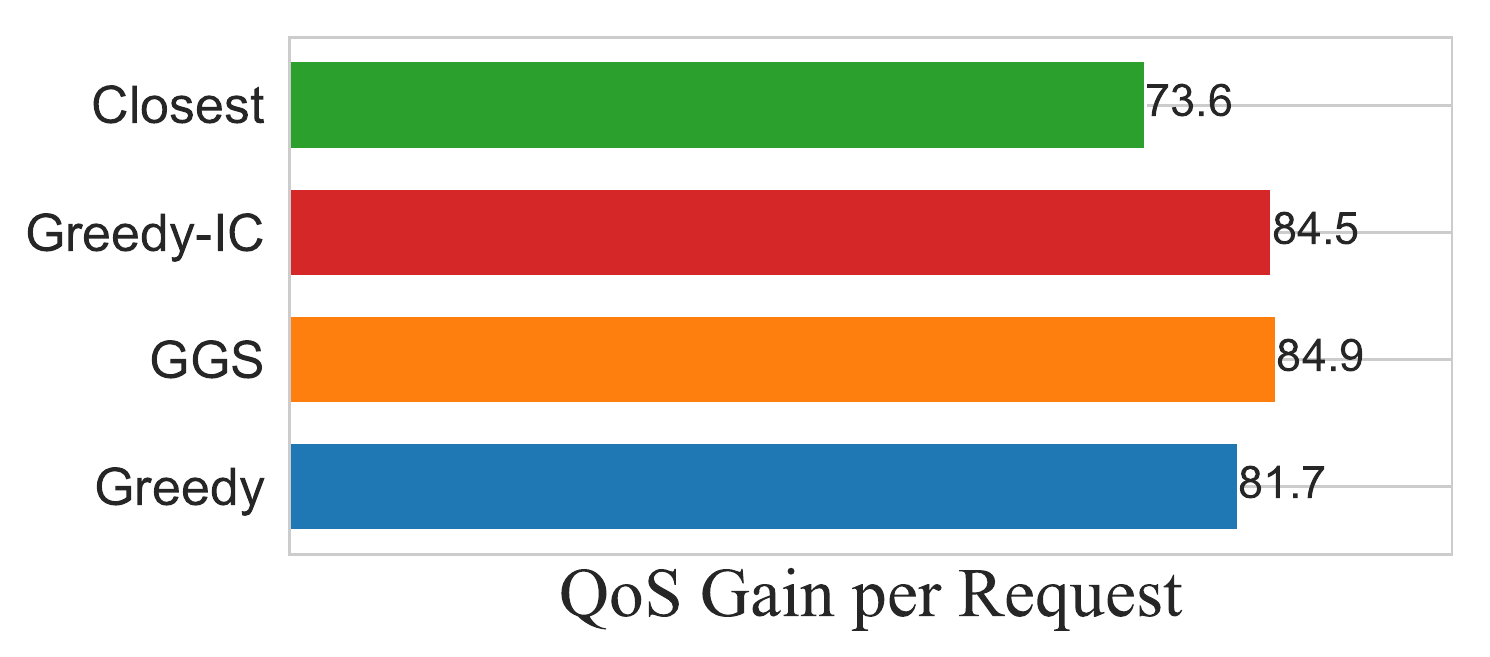}
}
\subfigure[Consistency
\label{tracediff}]{
\centering
\includegraphics[width=0.3\textwidth]{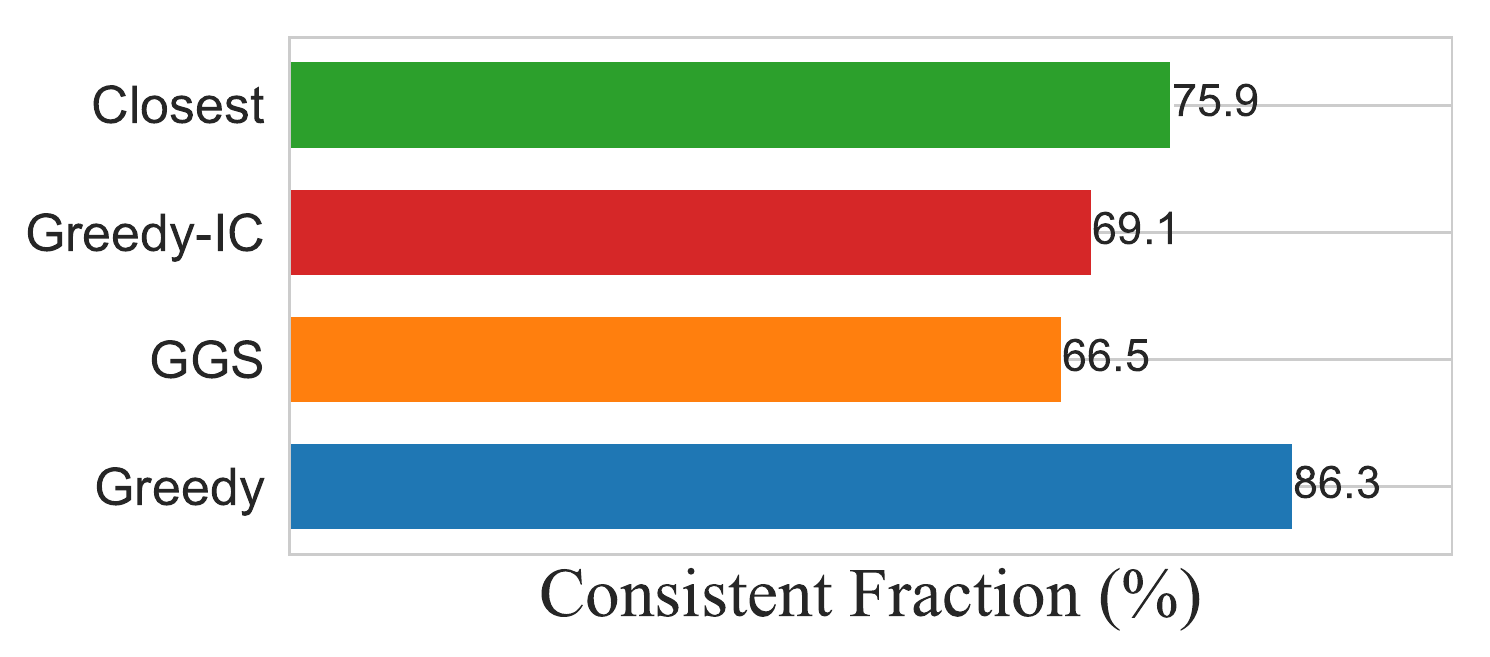}
}
\subfigure[Cache Miss Ratio
\label{tracemiss}]{
\centering
\includegraphics[width=0.3\textwidth]{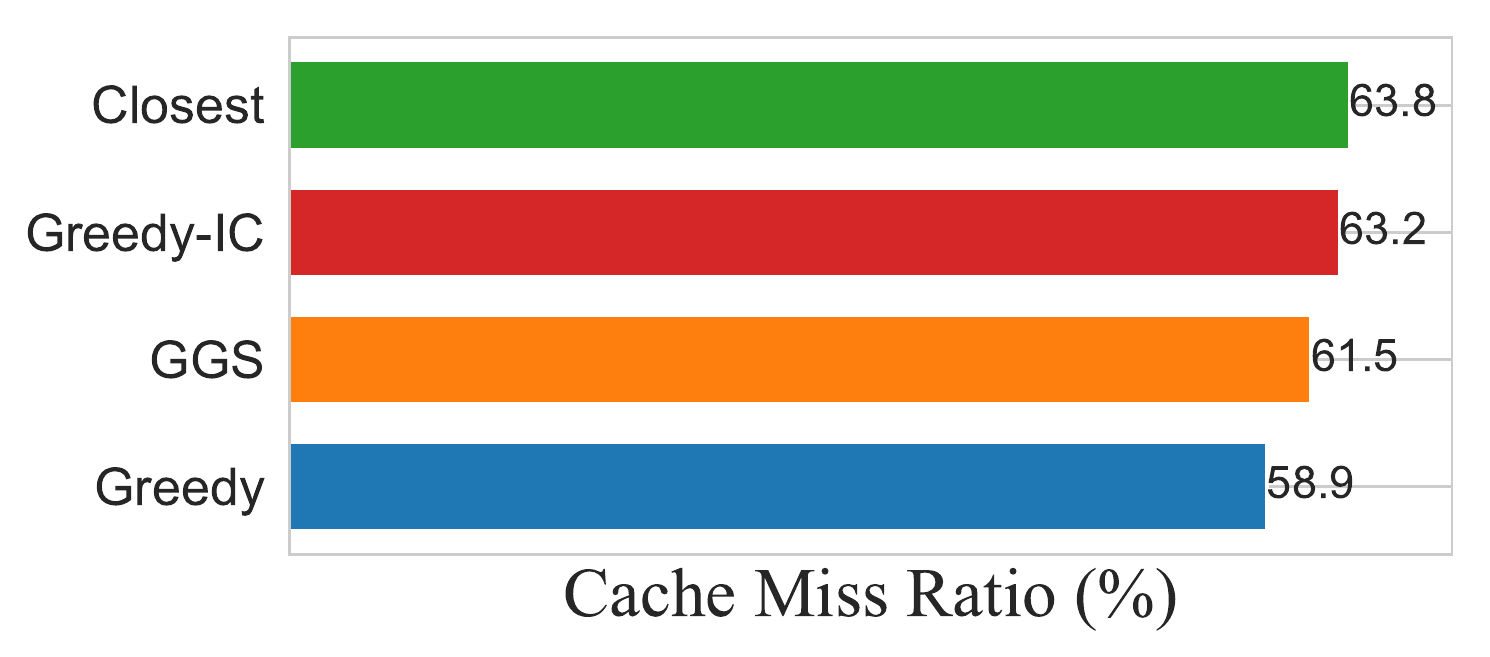}
}
\caption{\emph{Scenario 2: }we carry out a trace-driven evaluation using a short video request trace collected from real world, which reflects the dynamics of user requests. \emph{Greedy} performs more consistent without losing much QoS.
\label{trace}}
\end{figure*}

\subsubsection*{Baselines and Performance Metrics}
We name \Cref{greedy} as \emph{Greedy}, and compare it with three baselines: 1) \emph{GGS}, the generalized Gale-Shapley algorithm presented by Akamai\cite{akamai2012} that solves UTA problem algorithmically, but without considering consistency.  2) \emph{Greedy-IC}, the inconsistent version of \eqref{p1} whose objective function without $W(\mathbf{X})$ and solved by \Cref{greedy}. 3) \emph{Closest}, the algorithm that considers neither load balancing nor consistency. We assume the maximal length of the IP list to be 16\footnote{The maximal length of a DNS packet is 512B \cite{512}, and the size of IP address entry for IPv4/IPv6 is 16B/28B. The DNS response in one packet can consist of about 25 IPv4/14 IPv6 address at most. }.
The performance is evaluated by the following metrics:

\textbf{QoS gain per request: } we calculate average QoS gain per user request, which indicates how near user traffic can be served by a CBS.

\textbf{Consistent fraction: }each time the traffic allocation policy is recalculated, we record the fraction of "consistent" allocation, i.e.
 the fraction of user traffic that is allocated to a CBS that serves traffic under the same domain name in the previous traffic allocation policy.
 
\textbf{Cache miss ratio: }the average cache miss ratio of the mobile network during each time interval, which reflects the damage of "inconsistent" allocation.

\textbf{Cache hit ratio drop (CHRD): }as is shown in \Cref{rchr}, the real-time average cache hit ratio of all the caches suffers an immediate, drastic decrease after the regeneration of traffic allocation policy due to "inconsistent" allocation, as many contents need to be replaced. The ratio then increases and becomes stable. We record this instant drop of the cache hit ratio, which shows the harm of inconsistent traffic allocation.

\subsection{Scenario 1: Randomly-shuffled Workload}
In this scenario, we simulate the dynamics of content popularity.
Let there be 50 domain names with each domain name having 10,000 contents. The popularity of domain names and contents under each domain follow Zipf distribution. Zipf exponent of domain names and contents of each domain name are $\alpha_S=0.8$ and $\alpha_C=1.5$, because we think popularity intra-site skewness of contents is more prominent than the inter-site one. Considering that few regions with high population density while most of the regions have small population \cite{population}, we use a very flat Zipf distribution to approximate the traffic volume from UGs, with $\alpha_U=0.5$. 
We let \verb|CAPACITY|$=1200$, and set a total user traffic volume \verb|RATE|$=0.9\times$\verb|CAPACITY| (requests/s). Specifically, the traffic volume in a UG $u$ for domain name $s$ is expected to be \verb|RATE|$\times p_U(u)\times p_S(s)$, where $p_U$ and $p_S$ are the Zipf probability distribution functions of UG and domain name, respectively.
Then we sort the domain names by popularity and partition them into sections with equal length of \{5, 10, 25, 50\} domain names. The popularity of the domain names is randomly shuffled within each section, which may lead to traffic increase for some flows but decrease for others. The longer a section is, the higher popularity variance domain names in it have, and thus the more drastic popularity would change. The simulation lasts for 5,000 seconds and the time interval is 100 seconds, which means we shuffle the popularity and regenerate allocation policy every 100 seconds. Therefore, we have 50 results in each run. We record the real-time cache hit ratio every 2 seconds, in order to detect the immediate cache hit ratio drop. Before logging the results, we ran for another 500 seconds to warm up.

\textbf{Results. }Results are in \Cref{dynamic1}. \emph{Closest} yields least GoS gain because most of its requests are rejected by the CBS they are allocated to, which shows that if we don't consider load balancing, even when the total traffic volume remains the same, CBSs are likely to be overburdened. QoS gain of other algorithms is similar. In terms of consistency, when the content popularity sustains different degrees of oscillations, total QoS gain changes little (no more than $8\%$), but the traffic allocation policy is likely to change a lot. \emph{GSS} and \emph{Greedy-IC} can bring about up to nearly 40\% "inconsistent" allocation, while \emph{Greedy} no more than 10\%. It shows that \emph{GSS} and \emph{Greedy-IC} are more sensitive to the popularity changes than \emph{Greedy} and \emph{Closest}, which makes \emph{Greedy} yield least cache miss ratio. Although \emph{Closest} performs also well in terms of consistency, its total QoS gain is unsatisfactory. We present the average CHRD in \Cref{chrd} as well. It shows that the consistency of \emph{Greedy} results in its better performance in cache hit ratio oscillation, because it causes less content replacement, which validates the effectiveness of our formulation \eqref{p1}. 

\subsection{Scenario 2: Trace-driven Evaluation}
In this scenario, we use the dataset from \cite{trace} that records YouTube requests arising from the wired campus network. The trace we used lasts for 14 days in Feb. 2008, with 611,630 user requests from 6,670 anonymous users and 303,190 contents. We calculate the number of requests in each time interval as the total traffic volume \verb|RATE| for traffic allocation policy calculation in the algorithms. To avoid the influence of unpopular contents that do not even receive a second request, we screen out 20,000 most popular contents, and hash them into 20 domain names. The real request trace reflects both the user request patterns and the dynamics of traffic volume as well as content popularity. We divide all the users into 10 groups in order to partition the user requests into different flows according to the tuple\emph{(user group, domain name)}.
We let Time interval be 4 hours, and \verb|CAPACITY| be $80\%$ of the highest traffic volume among all the time intervals. We design the scenario this way as CBS is expected to serve a certain amount of user traffic during peak hours. 

\textbf{Results. }As shown in \Cref{trace}, we find that results in trace-driven evaluations show similar tendency with those in numerical evaluations. \emph{Greedy} increases the allocation consistency while trying to avoid much of QoS loss, which results in less cache miss ratio and thus saves more backhaul bandwidth. We think one of the reasons why the difference among algorithms is less clear than scenario 1 could be that the content popularity indicated in the trace is severely skewed. Specifically, among all the 303,190 contents there are only about 200 of them popular enough to be requested repeatedly in multiple time intervals. Therefore, even consistent traffic allocation policy may suffer from severe cache miss. The comparison with \emph{Closest} shows if load balancing is not considered when allocating user traffic, a large amount of traffic is likely to overburden popular CBSs. In such cases when QoS is severely damaged, pursuing consistency is useless.

\section{Conclusion}
\label{conclusion}
In this paper, we focus on the user-traffic allocation and load balancing in DNS-based mobile edge caching and coping with the two challenges posed by the practical implementation of caches in BSs. One is the granularity of DNS-based traffic scheduling, and the other is the need for consistent allocation policy under dynamic user-traffic.
We formulate the user-traffic allocation (UTA) problem in DNS-based mobile edge caching, which aims at maximizing QoS gain and allocation consistency as well as maintaining load balance. We then prove that the problem is equivalent to maximizing monotone submodular function that subjects to matroid constraints. A simple greedy algorithm is presented to solve this problem within 3/4 of the optimal solution. Extensive evaluations under both numerical and trace-driven situations show that the algorithm yields more consistent traffic allocation policy and thus results in less cache miss ratio and more balanced server load without losing much of QoS gain.



\bibliographystyle{IEEEtran}
\bibliography{cdn-hash}

\appendix

\subsection{Properties of~\eqref{p1}}
\label{detailed}


\subsubsection*{Ground Set}
Let element $f^k_{ij}$ represents the event that mapping unit $\langle i, k\rangle$ is assigned to CBS $j$, and thus the ground set in our UTA problem can be defined as:
\begin{equation}
\label{ground}
E = \{f^k_{ij}\mid j\in\mathbf{M}, i\in\mathbf{F}, k\in\mathbf{K_i}\}
\end{equation}

For a subset $A\subseteq E$, whether $f^k_{ij}$ is in $A$ depends on whether $x^k_{ij}=1$, thus a \emph{one-to-one correspondence} is achieved between a subset $A\subseteq E$ and the solution $\mathbf{X}$ to \eqref{p1}.

\subsubsection*{Constraints and Feasible Solutions}
\label{cons_sec}
Every element of $2^{E}$ (the power set of $E$) corresponds to a set of solutions to \eqref{p1}, and the set defined by the constraints in \eqref{p1} is no exception.

\begin{pro}
\label{cons}
Let $P_j=\{f^k_{ij}\mid i\in\mathbf{F}, k\in\mathbf{K_i}\}$, $Q_i^k=\{f^k_{ij}\mid j\in\mathbf{M}\}$ and $T^k_{ij}=\{f^k_{ij}\mid i\in\mathbf{F},j\in\mathbf{M},k\in\mathbf{K_i}\}$. The constraints in \eqref{p1} are equivalent to $\mathcal{I}$, where
\begin{subequations}
\label{matroid_def}
\begin{flalign}
\mathcal{I}_b=\{A\subseteq E\mid |A\cap P_j|\leq \lfloor \frac{c_j}{\lambda_0}\rfloor\}\label{capacons}\\
\mathcal{I}_c=\{A\subseteq E\mid |A\cap Q^k_i|\leq 1\}\\
\mathcal{I}_d=\{A\subseteq E\mid |A\cap T^k_{ij}|\leq t_{ij}\}\\
\mathcal{I}=\mathcal{I}_b \cap \mathcal{I}_c \cap \mathcal{I}_d
\end{flalign}
\end{subequations}
\end{pro}

\begin{proof}[Proof]
Suppose $A$ and $\mathbf{X}$ are equivalent, which means for every $j\in\mathbf{M}$, $i\in\mathbf{F}$, and $i\in\mathbf{K_i}$:
$f^k_{ij}\in A \Leftrightarrow x^k_{ij}=1$.

The sum of some $x^k_{ij}$ indicates that how many variables equal to 1 in all of them, which is exactly the cardinality of the intersection of $A$ and another set. For example, $$\sum_{j\in\mathbf{M}}x^k_{ij}=|A\cap Q^k_i|$$. Therefore, we have:
$$L_j(\mathbf{X})\leq c_j \Leftrightarrow \lambda_0\sum_{i\in\mathbf{F}}\sum_{k\in\mathbf{K_i}}x^k_{ij}\leq c_j
\Leftrightarrow \sum_{i\in\mathbf{F}}\sum_{k\in\mathbf{K_i}}x^k_{ij}\leq \lfloor \frac{c_j}{\lambda_0}\rfloor$$
 and $$\sum_{j\in\mathbf{M}}x^k_{ij}\leq 1 \Leftrightarrow |A\cap Q^k_i|\leq 1
$$.
Due to the discreteness of set, all the solutions to \eqref{p1} denoted by a set $A$ inherently satisfies \eqref{LM:4}. Thereinto, all the solutions satisfy \eqref{LM:2} are in $\mathcal{I}_b$, while that satisfies \eqref{LM:3} are in $\mathcal{I}_c$ and that satisfies \eqref{LM:5} are in $\mathcal{I}_d$, for which $\mathcal{I}=\mathcal{I}_b \cap \mathcal{I}_c\cap \mathcal{I}_d$ represents all the feasible solutions to \eqref{p1}, i.e.\ the constraints of \eqref{p1}.
\end{proof}

The tuple $(E, \mathcal{I})$ contains the ground set $E$ and the constraints $\mathcal{I}\in 2^E$. \Cref{matroid} shows that the tuple is a matroid\cite{matroid1} and therefore \eqref{p1} has matroid constraints.

\begin{pro}
\label{matroid}
$\mathcal{M}=(E, \mathcal{I})$ is a matroid with the definition of $E$ and $\mathcal{I}$ in \eqref{ground} and \eqref{matroid_def} respectively.
\end{pro}
\begin{proof}[Proof]
Review the definition of partition matroid: \emph{Partition matroid} is a typical instance of matroids. In a partition matroid, the ground set $E$ is partitioned into disjoint sets $E_1$, $E_2$,...,$E_l$ and 
$\mathcal{I} = \{A\subseteq E\mid |A\cap E_i|\leq \beta_i, \forall i=1,...,l\}$, for constant parameters $\beta_1, \beta_2,..., \beta_l$ \cite{partionmatroid}.

Likewise, $\{P_j\}$, $\{Q^k_i\}$ and $\{T^k_{ij}\}$ are three different partitions of the ground set $E$, and thus $\mathcal{M}_b=(E, \mathcal{I}_b)$, $\mathcal{M}_c=(E, \mathcal{I}_c)$ and $\mathcal{M}_d=(E, \mathcal{I}_d)$ are three partition matroids. $\mathcal{M}=(E, \mathcal{I})$ can be regarded as the intersection of $\mathcal{M}_b$, $\mathcal{M}_c$ and $\mathcal{M}_d$, which, according to \cite{matroid1}, is a matroid as well.
\end{proof}

\begin{cor}
\label{matroid_theo}
The constraints in \eqref{p1} are matroid constraints, and they are equivalent to the matroid $\mathcal{M}=(E, \mathcal{I})$.
\end{cor}

\subsubsection*{Objective Function}
\label{object_sec}
Due to the one-to-one correspondence between $\mathbf{X}$ and $A$, we can define the objective function of \eqref{p1} accordingly as a set function \cite{setfunction} $F':2^E \rightarrow \mathbb{R}$.
\begin{pro}
\label{objective}
The objective function $F(\cdot)$ of \eqref{p1} is equivalent to the set function $F'(\cdot)$:
\begin{equation}
\label{submodular_def}
F'(A) = \mu_1 G'(A)+\mu_2 B'(A)+\mu_3 W'(A)
\end{equation}
,where
\begin{gather}
\nonumber G'(A)=\sum_{f^k_{ij}\in A}\lambda_0g_{ij}\\
\nonumber B'(A)=-\sum_{j\in \mathbf{M}}(c_j-\lambda_0 |A\cap P_j|)^2\\
\nonumber W'(A)=\sum_{f^k_{ij}\in A}w_{ij}
\end{gather}

\end{pro}
\begin{proof}[Proof]
Suppose $A$ and $\mathbf{X}$ are equivalent, which means for every $j\in\mathbf{M}$, $i\in\mathbf{F}$, and $i\in\mathbf{K_i}$, $f^k_{ij}\in A \Leftrightarrow x^k_{ij}=1$.

$$G(\mathbf{X})=\sum_{j\in\mathbf{M}}\sum_{i\in\mathbf{F}}\sum_{k\in\mathbf{K_i}}\lambda_0g_{ij}x^k_{ij}$$ indicates that $G(\mathbf{X})$ is added by $\lambda_0g_{ij}$ if $x^k_{ij}=1$, for which $G(\mathbf{X})=G'(A)$.

Similarly, we have $$W(\mathbf{X})=\sum_{j\in\mathbf{M}}\sum_{i\in\mathbf{F}}\sum_{k\in\mathbf{K_i}}w_{ij}x^k_{ij}=\sum_{f^k_{ij}\in A}w_{ij}=W'(A)$$.

It's noted that $\sum_{i\in\mathbf{F}}\sum_{k\in\mathbf{K_i}}x^k_{ij}$ quantifies how many variables $x^k_{ij}=1$ when $j$ is given, which is equal to the value $|A\cap P_j|$. Therefore, we have $$L_j(\mathbf{X})=\lambda_0\sum_{i\in\mathbf{F}}\sum_{k\in\mathbf{K_i}}x^k_{ij}=\lambda_0|A\cap P_j|$$, and thus $$B(\mathbf{X})=-\sum_{j\in\mathbf{M}}(c_j-L_j(\mathbf{X}))^2
=-\sum_{j\in \mathbf{M}}(c_j-\lambda_0 |A\cap P_j|)^2=B'(A)
$$.

Finally, we come to a conclusion that the objective function of \eqref{p1}, $$F(\mathbf{X})=\mu_1 G(\mathbf{X})+\mu_2 B(\mathbf{X})+\mu_3 W(\mathbf{X})$$, equals to $F'(A)$.
\end{proof}

\Cref{submodular} shows that the equivalent objective function \eqref{submodular_def} is a monotone submodular function. Therefore, \eqref{p1} can be seen as maximizing a monotone submodular function, with regards to the constraints.

\begin{pro}
\label{submodular}
$F(\cdot)$ defined in \eqref{submodular_def} is a monotone submodular function.
\end{pro}
\begin{proof}[Proof]
For simplicity, we use $F_A(i)$ to denote the marginal value $F(A\cup\{i\})-F(A)$.
A set function is \emph{monotone} if $\forall A\subseteq B\subseteq E, F(A)\leq F(B)$\cite{submodular1}.

(\emph{Monotonicity}) For any $A \subseteq E$ and $f^k_{ij}\in E\backslash A$, since $g_{ij}\ge 0$ (see \Cref{optgoal}), $w_{ij}\ge 0$ (see \Cref{optgoal}) and $\lambda_0|A\cap P_j|\le c_j$ (see \eqref{capacons}), we have
$$
F'_A(f^k_{ij})=\mu_1 \lambda_0 g_{ij}+\mu_2 \cdot 2\lambda_0(c_j-\lambda_0|A\cap P_j|)+\mu_3 w_{ij}\geq 0
$$

(\emph{Submodularity}) For any $B_1\subseteq B_2 \subseteq E$, if we take out the elements from $B_2\backslash B_1$ one by one and add them to $B_1$, the value of $F'(\cdot)$ will not decrease as $F'_A(f^k_{ij})\geq 0$ for any $A$, for which $F'(B_2)\geq F'(B_1)$. $F'(\cdot)$ is hence a monotone function. 

For all $A\subseteq B\subseteq E$ and all $f^k_{ij}\in E\backslash B$,
$$
F_A(f^k_{ij})-F_B(f^k_{ij})=
\mu_2 \cdot 2{\lambda_0}^2(|B\cap P_j|-|A\cap P_j|)
$$. We have $|A\cap P_j|\leq|B\cap P_j|$ because $A\subseteq B$, and finally we have $F_A(f^k_{ij})\geq F_B(f^k_{ij})$. 
Since a set function is \emph{submodular} if $F_A(i)\geq F_B(i)$ for all $A\subseteq B\subseteq E$ and all $i\in E\backslash B$ \cite{submodular1}, $F(\cdot)$ is a submodular function.
\end{proof}

\begin{cor}
\label{submodular_theo}
The objective function of \eqref{p1} is equivalent to a monotone submodular function, namely, $F(\cdot)$, defined in \eqref{submodular_def}.
\end{cor}

\subsubsection*{Equivalent model}
In conclusion, the formulation of the UTA problem \eqref{p1} is equivalent to the following model:
\begin{subequations}
\label{p3}
\begin{alignat}{2}
\max\quad & F(A)\label{MM:1}\\
\mbox{s.t.}\quad & A\in\mathcal{I}\label{MM:2}
\end{alignat}
\end{subequations}
\eqref{matroid_def} and \eqref{submodular_def} define the constraints (i.e.\ a set of feasible functions) $\mathcal{I}$, and the objective function $F(\cdot)$, in which $\mathcal{M}=(E, \mathcal{I})$ is a matroid and $F(\cdot)$ is a monotone submodular function.

\subsection{Optimality of \Cref{greedy}}
\label{optimality}

\cite{greedy} proves that when the matroid constraint $\mathcal{I}$ can be written as the intersection of $P$ matroids, i.e.\ $\mathcal{I}=\bigcap_{p=1}^P\mathcal{I}_p$, the greedy algorithm yields a tight approximation ratio of $\frac{P}{P+1}$.
Let $A^*$ be the optimal solution of \eqref{p1} and $A^G$ be the output of \Cref{greedy}. The approximation ratio means
$$
\frac{F'(A^*)-F'(A^G)}{F'(A^*)-F'(\emptyset)}\leq \frac{P}{P+1}
$$
. In our problem, $P=3$ as $\mathcal{I} = \mathcal{I}_b \cap \mathcal{I}_c \cap \mathcal{I}_d $, for which the greedy algorithm is supposed to yield a 3/4 approximation.


\end{document}